\documentclass{amsart}
\usepackage{latexsym}
\usepackage{amsmath,amssymb,mathrsfs,pictex} 
\usepackage{amsmath,amsthm,amsfonts,amscd,eucal}
\usepackage{graphicx}
\usepackage{epsfig}
\usepackage{epstopdf} 
\usepackage{cancel}

\numberwithin{equation}{section}

\def\cb{{\mathcal B}}

\def\cd{{\mathcal D}}

\def\cf{{\mathcal F}}

\def\ch{{\mathcal H}}

\def\cs{{\mathcal S}}

\def\br{{\mathbb R}}
\def\bs{{\mathbb S}}

\def\bz{{\mathbb Z}}

\def\a{\alpha}
\def\b{\beta}
\def\g{\gamma}        
\def\d{\delta}        \def\D{\Delta}
\def\eps{\varepsilon} 
     
\def\z{\zeta}

\def\l{\lambda}       
\def\m{\mu}
\def\n{\nu}

\def\r{\rho}
\def\s{\sigma}

\def\om{\omega}        \def\Om{\Omega}

\newtheorem{Thm}{Theorem}[section]

\newtheorem{Prop}[Thm]{Proposition}

\theoremstyle{definition}

\newtheorem{Rem}[Thm]{Remark} 
\theoremstyle{remark}

\def\di{\mathop{\mathop{\rm d}}\!}

\newcommand{\ty}[1]{\mathop{\rm {#1}}}

\def\supp{\mathop{\rm supp}}
\def\ds{\mathop{\bcancel{\rm d}}\!}

\newcommand{\nn}{\nonumber}

\begin{document}

\title[thermodynamics of open systems]
{a proposal for the thermodynamics of certain open systems}
\author{Francesco Fidaleo}
\address{Dipartimento di Matematica,
Universit\`{a} di Roma Tor Vergata, 
Via della Ricerca Scientifica 1, 00133 Rome, Italiy} 
\email{fidaleo@mat.uniroma2.it} 
\author{Stefano Viaggiu}
\email{viaggiu@mat.uniroma2.it}


\date{\today}

\begin{abstract}
Motivated by the fact that the (inverse) temperature might be a function of the energy levels in the Planck distribution $n_\eps=\frac1{\z^{-1}e^{\b(\eps)\eps}-1}$ for the occupation number $n_\eps$ of the level $\eps$, we show that it can be naturally achieved by imposing the constraint concerning the conservation of a weighted sum $\sum_{\eps}f(\eps)\eps n_\eps$, with a fixed positive weight function $f$, of the contributions of the single energy levels occupation in the Microcanonical Ensemble scheme, obtaining $\b(\eps)\propto f(\eps)$.
This immediately addresses the possibility that also a weighted sum $\sum_{\eps}g(\eps)n_\eps$
of the particles occupation number is conserved, having as a consequence that the chemical potential might be a function of the energy levels of the system as well. 
This scheme leads to a thermodynamics of open systems in the following way: 
\vskip.2cm
\noindent
{\it the equilibrium is reached when the entropy function is maximised under the constraints that some weighed sums of occupation of the energy levels and the occupation numbers are conserved}. 
\vskip.2cm
\noindent
The standard case of isolated systems corresponds to the weight functions being trivial (i.e. $f, g$ are identically 1).
For such open systems, new and unexpected phenomena which might happen in nature can appear, like the Bose Einstein Condensation in excited levels. The ideas outlined in the present paper may provide a new approach for the treatment of the irreversible thermodynamics.\\
\vskip.1cm
\end{abstract}

\maketitle

\section{introduction}

The possibility that the (inverse) temperature can be a function of the energy levels of the system appeared in \cite{AFQ} as {\it Local Equilibrium} even if, perhaps, it was considered in previous studies. Recently, in \cite{AFQ1} it has been investigated the connection of the Local Equilibrium with the principle of detailed balance for "small" open systems interacting with a "huge" reservoir.

The Local Equilibrium simply means that, in the celebrated Planck formula for the occupation numbers of Bose particles
\begin{equation}
\label{plq}
n_\eps=\frac1{\z^{-1}e^{\b\eps}-1}\,,\quad \eps\in\,\text{the set of energy levels of the system}\,,
\end{equation}
the inverse temperature is supposed to be a function of $\eps$: $\b=\b(\eps)$. Here, $\z$ is the fugacity, and $q=0,\pm1$ correspond to the Boltzmann and Bose/Fermi cases.

A rigorous approach to the Local Equilibrium in terms of the Kubo-Martin-Schwinger (KMS for short) boundary condition can be carried on essentially for systems with finite degrees of freedom and/or systems confined in a bounded spatial region for which, typically, the observables are described by all bounded operators $\cb(\ch)$ on the separable Hilbert space $\ch$, and the hamiltonian of the system is semibounded with compact resolvent with a good behaviour of the asymptotic of the eigenvalues.

For the sake of completeness, we note that a possible way to generalise the Local Equilibrium Principle is to look at the Arveson spectrum of the one parameter group of automorphisms $\{\a_t\}_{t\in\br}$ describing the evolution of the system. Under very natural assumptions, it is also quite well known that for a KMS state $\om$, the Arveson spectrum coincides with the spectrum of $H_\om$ (e.g. \cite{BF}), $H_\om$ being proportional to the logarithm of the modular operator associated to $\om$ by Tomita-Takesaki Theory. This approach has been considered in \cite{dC} in the investigation of the so called spectrally passive states, providing a possible bridge between the Local Equilibrium Principle for general systems and the KMS boundary condition. 

Unfortunately, all physical systems with compact resolvent hamiltonians lead to type $\ty{I}$ factors, whereas it is well known that most of the nontrivial physical models arising from the thermodynamics produce von Neumann algebras of type $\ty{II_1}$ (corresponding to infinite temperature) and $\ty{III_\l}$, $\l\in(0,1]$, even for systems of non interacting particles, see e.g. 
\cite{BF0, BF, BR} and the references therein.

To provide physically relevant examples encompassing the more realistic situation described above, in \cite{AF} the Local Equilibrium is directly defined in terms of gauge invariant quasi free states of CCR algebras, extending the definition to the $q$-deformed ones, $-1\leq q\leq 1$. We then recover the previous Bose case whenever $q=1$, and include the Fermi and the Boltzmann cases $q=-1,0$ respectively. 

After introducing a parameter $\l$ playing mathematically the role of a {\it chemical potential}, for such very natural situation the analogous of \eqref{plq} is given by
\begin{equation}
\label{0plq}
n_\eps=\frac1{e^{\b(\eps)\eps-\l}-q}
\end{equation}
plus, possibly, a distribution supported on the subset $\{{\bf p}\mid\b(\eps({\bf p}))\eps({\bf p})=0\}$ in momentum space. It is proven that such a distributional term in the occupation number density can happen only for positive $q$, and describes the condensation of $q$-particles, excluded for the Fermi-like ones $-1\leq q<0$, and allowed for the Bose-like ones $0<q\leq1$. As it is well known, we again recover that the condensation is forbidden also for the classical situation corresponding to the Boltzmann case $q=0$. 

The Local Equilibrium implies some interesting phenomena like these for which the Bose Einstein Condensation (BEC for short) can take place also on excited energy levels. In addition, this new approach based on the Distribution Theory allows the construction of quasi free states exhibiting the BEC which are completely unknown in literature even in the standard equilibrium situation. 

From the analysis of the Local Equilibrium briefly described above, it emerges that the assumption that the inverse temperature can be a function of the energy levels of the system under consideration, and the introduction of the chemical potential, are made {\it ad hoc} without any further justification. It is then natural to address the question concerning the possibility to recover the framework arising by the Local Equilibrium, from more reasonable basic assumptions which appear physically meaningful. This is precisely the main goal of the present paper.

For such a purpose, we consider the {\it Microcanonical Ensemble} for which the thermodynamical properties of a system made of 
a number of particles of the order of the Avogadro Number
$N_A\sim 10^{23}$ is encoded in the Entropy Functional $S$. Such a functional takes into account the complexity of so huge systems by counting, in an appropriate way, all possible microscopical configurations reproducing the same macroscopical one determined by fixing the total number of particles and energy. To be more precise, fix a system whose hamiltonian $H$ is positive and has compact resolvent:
\begin{equation}
\label{plq3}
H=\sum_{\eps_i\in\s(H)}\eps_i P_{\eps_i}\,.
\end{equation}
Consider the occupation numbers of particles $n_i$ of the energy level $\eps_i$. The Entropy Functional $S$ is defined as
\begin{equation}
\label{plq11}
S(\{n_i\})=k\ln W(\{n_i\})\,,
\end{equation}
where $k=k_B\approx1.3806488\times 10^{-23} JK^{-1}$ is the Boltzmann constant, and $W$ counts all possible microscopical configurations corresponding to the sequence of occupation numbers $\{n_i\}$. It depends on the statistics (Bose/Fermi or Boltzmann) to which obey the identical particles of the system under consideration. Its precise form is given in \eqref{entc} (where $g_i$ is the degeneracy of the level $i$, and $\n_i=n_i/g_i$) after approximating the factorials by the Stirling formula which is legitimated by the enormous number of particles composing the systems under consideration.

The standard equilibrium thermodynamics arises by maximising the entropy by fixing the total energy $E$ and the number of particles $N$, which simply corresponds to maximise the functional
$S(\{n_i\})$ with the constraints 
\begin{equation}
\label{plq1}
\sum_{i}\eps_in_i=E\,, \quad \sum_{i}n_i=N\,.
\end{equation}
The universally accepted approach briefly outlined above provides the starting point of the explanation of the thermodynamics with the ideas of Statistical Mechanics. The reader is referred to the standard textbooks (e.g. \cite{Hu, LL}) for the detailed explanation of the connections between thermodynamics and Statistical Mechanics, and for the natural consequences and various applications.

Summarising in simple words, all that outlined above corresponds to the {\it equilibrium thermodynamics of isolated systems}, where the equilibrium corresponds to maximise the entropy, and such an equilibrium should be reached imposing the conditions \eqref{plq1}. Then we can refer the above analysis of the equilibrium to isolated systems, that is those without any exchange of either energy or matter with the surrounding environment.

It is now natural to consider enormous systems whose complexity is still encoded in the Entropy Functional, but the equilibrium is determined not just by the conservation of the energy and matter, but by the conservation of certain weighted sum of the energy level  and occupation numbers like
\begin{equation}
\label{plq2}
\sum_{i}f(\eps_i)\eps_in_i=e\,, \quad \sum_{i}g(\eps_i)n_i=n\,,
\end{equation}
where $f,g$ are positive functions defined on the spectrum of the hamiltonian of the system \eqref{plq3}. The case $f=g=1$ reproduces the standard equilibrium thermodynamics of isolated systems.
In the present article, 
in searching the condition of equilibrium maximising the entropy,
we relax \eqref{plq1} replacing that with \eqref{plq2}. In so doing, the very simple output (cf. Theorem \ref{vincc1}) asserts that:\\
\vskip.1cm
\noindent
 {\it the maximum of the entropy \eqref{plq11} with the constraints \eqref{plq2} is reached whenever the occupation numbers have the values}
\begin{equation}
\label{3plq3}
n_{\eps_i}=\frac1{e^{b(f(\eps_i)\eps_i-mg(\eps_i))}-q}\,.
\end{equation}
\vskip.2cm
\noindent
Here, the Lagrange multipliers $b$, $bm$ assume the meaning of the generalised inverse temperature and the logarithm of the generalised fugacity, and can be computed in terms of $e,n$ (or equivalenty $E,N$) inserting \eqref{3plq3} in \eqref{plq2} (or \eqref{plq1}). Notice that \eqref{3plq3} coincides with \eqref{0plq} (with $\b(\eps)=bf(\eps)$ and $\l=bm$) whenever $g(\eps)=1$. 

Theorem \ref{vincc1} immediately leads to the fact that in the generalised Planck distribution \eqref{0plq} the inverse temperature might be a function of the energy levels of the system under consideration, and the appearance of a coefficient playing the role of the chemical potential, directly follow from a natural generalisation of the equilibrium principle applied to certain open systems described by the conditions \eqref{plq2} with $g$ identically 1. These correspond to the closed systems, that is those which can exchange only heat/energy with the environment, for which $\sum_{i}f(\eps_i)\eps_in_i$ and $\sum_{i}n_i$ are conserved quantities.

The conclusion of the previous analysis is simply that the complexity of certain open systems satisfying conditions \eqref {plq2} are still encoded in the entropy. As in the standard situation of isolated systems, the equilibrium or equivalently the maximum of the entropy, is reached for occupation numbers \eqref{3plq3}. It is expected that such a generalisation of the principle of the equilibrium might provide the thermodynamical properties of such quite general open systems which might be present in nature.
We also remark that a similar approach has been followed by C. Tsallis where the $q$-entropy considered here is replaced by Tsallis entropy, see e.g. \cite{BD, T, T1}.

In the present paper we also discuss the relations between the generalised (inverse) temperature $b$ and the chemical potential $m$, with the corresponding "physical objects"
$$
\frac1{T}:=\bigg(\frac{\partial S}{\partial U}\bigg)_{N,V}\,,\quad \frac{\m}{T}:=-\bigg(\frac{\partial S}{\partial N}\bigg)_{U,V}\,,
$$
$U, N, V$ being the internal energy, the average number of particles and the volume of the system under consideration.
 
 A section is devoted to the appearance of the BEC for the open systems under consideration generalising some results contained in \cite{AF}. 
 
 Finally, we briefly discuss the Boltzmann gas of free massive particles in a box of volume V in $\br^3$ with 
$$
f(\eps({\bf p}))=a\big(p_x^2+p_y^2+p_z^2\big)^{s/2}\,, \quad a>0\,, s>-2\,,
$$ 
and $g(\eps({\bf p}))=1$, for which all explicit calculations can be carried on.

\section{microcanonical ensemble}
\label{1mic1}

In order to justify the Local Equilibrium which allows the inverse temperature to be a function of the energy levels in the Planck formula, we deal with the Microcanonical Ensemble (cf. \cite{Hu, LL}) and the corresponding Entropy Functional. By using a natural generalisation, also the chemical potential is allowed to be a function of the energy levels of the model. For the sake of the completeness, we derive the formula of the $q$-entropy corresponding to exotic models relative to $q$-particles, $-1\leq q\leq 1$. 

As usual, we start from a system whose hamiltonian $H$ is a selfadjoint strictly positive matrix 
$$
H=\sum_{\eps_i\in\s(H)}\eps_i P_{\eps_i}
$$
uniquely characterised up to unitary equivalence, by the set $\{\eps_i\}$ of its eigenvalues and its degeneracy of the levels (i.e. the multiplicity) 
$$
g_i:=\text{dim}\text{R}(P_{\eps_i})\,.
$$
We can suppose equally well that $H$ is a densely defined positive selfadjoint unbounded operator with compact resolvent acting on an infinite dimensional Hilbert space, obtaining an extreme problem on an infinite dimensional space. As this technicality is not adding anything else to our analysis, we decide not to pursue such a generalisation. 

Suppose that $N$ indistinguishable particles are occupying the levels $\eps_i$ with occupation numbers $n_i$ under the obvious condition $N=\sum_in_i$.
According to the three cases Bose/Fermi and Boltzmann respectively, the number $W(\{n_i\})$ of such possible configurations is given by
$W(\{n_i\})=\prod_iw_i$ 
with
\begin{equation*}
w_i=\left\{
\begin{array}{ll}
\binom{n_i+g_i-1}{n_i} & \text{Bose}\,, \\
\binom{g_i}{n_i}  & \text{Fermi}\,, \\
\frac{g_i^{{n_i}}}{n_i! } & \text{Boltzmann}\,, \\
\end{array}
\right.
\end{equation*}
after dividing $W(\{n_i\})$ by $N!$ in the Boltzmann one, see e.g. Section 8.5 of \cite{Hu}. As usual, we suppose that all $g_i$ and $n_i$ go to infinity justifying the replacement of the factorials with their asymptotic by Stirling formula $m!\approx m^me^{-m}$, obtaining for the entropy
$S(\{\n_i\}):=\ln W(\{n_i\})$ (in the units for which $k_B=1$)
\begin{equation}
\label{entc}
S(\{\n_i\})=\left\{
\begin{array}{ll}
\sum_i g_i[\n_i\ln(1/\n_i+1)+\ln(1+\n_i)] & \text{Bose}\,, \\
\sum_i g_i[\n_i\ln(1/\n_i-1)-\ln(1-\n_i)]  & \text{Fermi}\,, \\
\sum_i g_i\n_i(1-\ln\n_i) & \text{Boltzmann}\,.
\end{array}
\right.
\end{equation}
Here, we have put $\n_i:=n_i/g_i$. 

The entropies given in \eqref{entc} for the Bose/Fermi and Boltzmann alternative can be considered as particular cases of 
of the $q$-{\it entropy} defined for $q\in[-1,0)\cup(0,1]$,
\begin{equation}
\label{0entc}
S_q(\{\n_i\}):=\sum_i g_i\bigg[\frac{(1+q\n_i)}{q}\ln(1+q\n_i)-\n_i\ln\n_i\bigg]\,.
\end{equation}
In fact, the Bose/Fermi cases correspond to the evaluation of $S_q$ for $q=\pm 1$, respectively:
$$
S_{+1/-1}(\{\n_i\})=S_\text{Bose/Fermi}(\{\n_i\})\,.
$$
Concerning the Boltzmann case, we get
$$
\lim_{q\to0}S_q(\{\n_i\})=S_\text{Boltzmann}(\{\n_i\})\,,
$$
pointwise in the variables $\n_i$, and uniformly on all bounded subsets (in the variables $\{\n_i\}$). 

The $q$-entropy could be computed as before by using the so called $q$-deformed statistics, firstly considered in \cite{Fiv}, arising from the $q$-deformed Canonical Commutation Relations. We refer the reader to (14) in \cite{LS} for a similar Entropy Functional still arising from a class of deformed commutation relations. We decide not to pursue more these points because, as usual, we take \eqref{0entc} as the definition of the Entropy Functional.

To avoid unpleasant situations, we fix two strictly positive functions $f$, $g$ on the spectrum $\{\eps_i\}$ of the hamiltonian $H$. However, concerning the continuum case, we can allow the functions
$f$, $g$ to be zero on a negligible subset w.r.t. the measure determined by the resolution of the identity \eqref{1en112} of the one particle hamiltonian, see e.g. Section \ref{BZ}.

The main point of the present paper is to consider the extreme problem for the Entropy Functional \eqref{0entc} with the constraints
\begin{equation}
\label{vinc}
\sum_i f(\eps_i)\eps_in_i=e\,,\quad \sum_i g(\eps_i)n_i=n\,.
\end{equation} 
Here, $e$, $n$ correspond to the weighted sums involving the number of particles and the energy of the system which, in our thermodynamical scheme, are considered as conserved quantities. They depend on the chosen functions $f$ and $g$, which are not explicitly mentioned to shorten the notation. Notice that we can recover the usual thermodynamics when they are identically $1$, obtaining $e=E$ the total energy of the system, and $n=N$ the total number of particles respectively. 
\begin{Thm}
\label{vincc1}
The values $\{\bar\n_i\}$ which maximise the $q$-entropy $S_q$ in \eqref{0entc} subjected to the constraints \eqref{vinc} are given by 
\begin{equation}
\label{vinc1}
\bar\n_i=\frac1{e^{b(f(\eps_i)\eps_i-mg(\eps_i))}-q}\,.
\end{equation} 
\end{Thm}
\begin{proof}
By the previous considerations, we can directly manage $S_q$ for generic $q$ and see that the obtained result still holds for $q=0$. Consider
$$
L((\{\n_i\},b,m):=S_q(\{\n_i\})-b\sum_ig_if(\eps_i)\eps_i\n_i + bm\sum_ig_ig(\eps_i)\n_i\,,
$$
where the form of the Lagrange multipliers $b$ and $-bm$ have been chosen in this way for physical motivations. We get
$$
\frac{\partial L}{\partial\n_i}=g_i\bigg[\ln\frac{1+q\n_i}{\n_i}-b\big(f(\eps_i)\eps_i-m g(\eps_i)\big)\bigg]\,.
$$
Solving the equations $\frac{\partial L}{\partial\n_i}=0$  w.r.t. the $\n_i$, we get \eqref{vinc1}. In addition,
$$
\frac{\partial^2 L}{\partial\n_i\partial\n_j}=-\d_{ij}\frac{g_i}{(1+q\n_i)\n_i}<0
$$
for $\bar\n_i>0$. In fact, it holds automatically true for $q\in[0,1]$. If $q\in[-1,0)$, then it holds true whenever $|q|\n_i<1$. But for $\n_i=\bar\n_i$ we get
$$
|q|\bar\n_i=\frac1{\frac{e^{b(f(\eps_i)\eps_i-mg(\eps_i))}}{|q|}+1}<1\,.
$$
\end{proof}
Notice that the Lagrange multipliers $b$,  $m$ are determined by inserting $\bar n_i=g_i\bar\n_i$ in \eqref{vinc}.

We first note that \eqref{vinc1} is nothing but the Planck occupation number obtained in a rigorous way in \cite{AF} by fixing the chemical potential (corresponding to the Grand Canonical Ensemble for the usual equilibrium thermodynamics) for the trivial function $g(\eps)=1$. Apart from the more rigorous approach, the method in the previous paper \cite{AF} takes into account even the appearance of the BEC. The approach carried out here has the conceptual advantage to explain in a very precise way the motivation for which the temperature is allowed to be a function of energy levels of the system: it is connected with the constraint in \eqref{vinc} associated to a weighted sum of the energies of the levels. On one hand, it reduces to the usual Planck occupation number provided also $f(\eps)=1$, identically. On the other hand, it explains that the choice of the introduction of the chemical potential in \cite{AF} (when it is independent on the temperature) is supported by the previous analysis and corresponds to the case $g(\eps)=1$. The other possible choice discussed in Section 7 of \cite{AF} is also allowed with the choice $g(\eps)\propto f(\eps)$.

\section{some thermodynamic quantities}
\label{tprr}

We discuss a possible comparison of the new thermodynamical variables $b$, $m$ naturally arising from our approach to thermodynamics of open systems, with the previous ones $\b:=1/kT$ and 
$\m$, being respectively the inverse of the temperature $T$ times the Boltzmann constant, and the chemical potential. With $S$ the entropy of the system, they are respectively defined as
\begin{equation}
\label{en00}
\b:=\frac{\partial S}{\partial E}\,,\quad \b\m:=-\frac{\partial S}{\partial N}\,.
\end{equation} 
Notice that, in our framework, $E, N$ are meant non just as conserved quantities, but just as the averaged internal energy and the particle number of the system according to the distribution \eqref{vinc1}. Obviously, these coincide with the conserved quantities in the standard equilibrium case.

The definition of the standard temperature $T$ and the chemical potential $\m$ also arise mathematically as Lagrange multipliers, even if those have a physical justification, see e.g. \cite{Hu, LL}.
In our more general framework, the {\it generalised inverse temperature} $b$ and the {\it generalised chemical potential} $m$ satisfy
\begin{equation}
\label{en0}
b:=\frac{\partial S}{\partial e}\,,\quad bm:=-\frac{\partial S}{\partial n}\,.
\end{equation} 
When $f=g=1$ in \eqref{vinc}, then $b=\b$ and $m=\m$, and we are recovering the standard thermodynamics. The energy $E$ and the number of particles $N$ are naturally given by
\begin{equation}
E=\sum_i \eps_ig_i\bar\n_i\,,\quad N=\sum_ig_i\bar\n_i\,.
\label{en02}
\end{equation}
Then by using \eqref{vinc1}, and taking into account that $b$ and $m$ are function of $e$ and $n$, we obtain
\begin{equation}
\label{en01}
E=E(e,n)\,,\quad N=N(e,n)\,.
\end{equation} 
Suppose now that such functions are smooth and invertible, at least locally. This might be not true in general when the system exhibits phase transitions, even if in this situation it is still possible to manage such a problem, see e.g. \cite{AF, BR, Hu, LL} and Theorem \ref{tcbb} below.
When we have invertibility and regularity, we get
\begin{equation}
\label{en1}
e=e(E,N)\,,\quad n=n(E,N)\,.
\end{equation} 
By using \eqref{en1} in \eqref{en00} and after taking into account the \eqref{en0}, we obtain
\begin{equation}
\label{en112}
\b=b\bigg(\frac{\partial e}{\partial E}-m\frac{\partial n}{\partial E}\bigg)\,,\quad
\b\m=b\bigg(m\frac{\partial n}{\partial N}-\frac{\partial e}{\partial N}\bigg)\,.
\end{equation} 
We can thus compute all standard thermodynamical quantities like the usual temperature and the usual chemical potential in terms of the energy and the number of particles of the system.

In the present paper, we mainly deal with the ideal gas of non interacting particles, whose one particle hamiltonian $h$ acting on the one particle (separable) Hilbert space is given by 
\begin{equation}
\label{1en112}
h=\int_{[0,+\infty)} \eps\di e(\eps)\,.
\end{equation}
We discuss in more detail the conditions such that the \eqref{en1}, and consequently the \eqref{en112}, hold true. 
After defining for $x:=b$, $y:=-bm$,
$$
H(x,y;\eps):=x\eps f(\eps)+yg(\eps)\,,
$$
for a fixed $q\in[-1,1]$ we assume that there exists a Borel measure $\m$ on the real line with $\supp(\m)\subset [0,+\infty)$ such that
\begin{equation}
\label{ennump}
E(x,y)=\int\frac{\eps\di\m(\eps)}{e^{H(x,y;\eps)}-q}\,,\quad N(x,y)=\int\frac{\di\m(\eps)}{e^{H(x,y;\eps)}-q}\,,
\end{equation}
\begin{equation}
\label{ennump1}
e(x,y)=\int\frac{\eps f(\eps)\di\m(\eps)}{e^{H(x,y;\eps)}-q}\,,\quad n(x,y)=\int\frac{g(\eps)\di\m(\eps)}{e^{H(x,y;\eps)}-q}\,.
\end{equation}
Here, $f,g$ are measurable functions which are non negative, almost surely w.r.t. $\m$. The measure $\m$ describes "the density of eigenvalues" of the hamiltonian in the thermodynamic limit and can be recovered by the resolution of the identity \eqref{1en112} for many concrete example. Its cumulative function $N_h(\eps):=\m(-\infty,\eps]$ is known as the {\it Integrated Density of the States}. Concerning the ideal gas with one particle hamiltonian $h=-\frac{\D}{2M}$, it is easily computed as in Section \ref{appp} by using the Fourier Transform. We also refer the reader
to \cite{F2} for a rigorous approach to the definition and the computation of the Integrated Density of the States for quite general models associated to infinite networks.

We now suppose that there exists values $(x_0,y_0)$ and a neighborhood $U\ni(x_0,y_0)$ such that for $(x,y)\in U$
\begin{itemize}
\item[(i)] $e^{H(x,y;\eps)}-q\geq a>0$, almost everywhere w.r.t. $\m$ (absence of the BEC);
\item[(ii)] all integrals in \eqref{ennump} and \eqref{ennump1} are meaningful: 
$$
E(x,y)\,, N(x,y)\,, e(x,y)\,, n(x,y)<+\infty\,;
$$
\item[(iii)] the following integrals, which will provide (up to a sign) the partial derivatives of $E,N$ w.r.t. $x,y$, are finite:
\begin{align*}
&\int\frac{\eps^2f(\eps)e^{H(x,y;\eps)}\di\m(\eps)}{\big(e^{H(x,y;\eps)}-q\big)^2}\,,\quad \int\frac{\eps g(\eps)e^{H(x,y;\eps)}\di\m(\eps)}{\big(e^{H(x,y;\eps)}-q\big)^2}<+\infty\,,\\
&\int\frac{\eps f(\eps)e^{H(x,y;\eps)}\di\m(\eps)}{\big(e^{H(x,y;\eps)}-q\big)^2}\,,\quad \int\frac{g(\eps)e^{H(x,y;\eps)}\di\m(\eps)}{\big(e^{H(x,y;\eps)}-q\big)^2}<+\infty\,;
\end{align*}
\item[(iv)] there exists a measurable function $K$ such that:
\begin{align*}
e^{H(x,y;\eps)}\leq K(\eps)&\,\,\text{almost everywhere w.r.t.}\,\, \m\,,\\
\eps^2f(\eps)K(\eps)\,,\eps g(\eps)K(\eps)&\,,\eps f(\eps)K(\eps)\,, g(\eps)K(\eps)\in L^1(\m)\,.
\end{align*}
\end{itemize}
We note that (i) above is added mainly for physical motivations. In fact, in presence of BEC the quantities $E(x,y)$, $N(x,y)$, $e(x,y)$, $n(x,y)$ might be not univocally determined by \eqref{ennump} and \eqref{ennump1} for the possible appearance of some delta distribution as explained before \eqref{appbec}.
\begin{Prop}
Assume that {\rm (i)-(iv)} above are satisfied, so that in particular, 
the quantities $E(x,y)\,, N(x,y)\,, e(x,y)\,, n(x,y)$ defined by \eqref{ennump}, \eqref{ennump1} are finite.

$$
\text{If}\,\,\int\!\!\!\!\int(\eps_1-\eps_2)\frac{\eps_1 f(\eps_1)g(\eps_2)e^{H(x_0,y_0;\eps_1)+H(x_0,y_0;\eps_2)}}{\big[\big(e^{H(x_0,y_0;\eps_1)}-q\big)\big(e^{H(x_0,y_0;\eps_2)}-q\big)\big]^2}\di\m(\eps_1)\di\m(\eps_2)\neq0\,,\quad\quad\quad
$$
then there exists a neighborhood $(x_0,y_0)\in V\subset U$ for which the vector function ${\bf F(x,y)}:=(E(x,y), N(x,y))$ is invertible with smooth inverse 
$$
{\bf F}^{-1}:{\bf F}(V)\rightarrow V\,.
$$
As a consequence, on $V$ the modified energy $e$ and the modified number $n$ are functions of the average energy and the average number of particles of the system:
$$
e(E,N)=\int\frac{\eps f(\eps)\di\m(\eps)}{e^{(H\circ {\bf F}^{-1})(E,V;\eps)}-q}\,,\quad n(E,N)=\int\frac{g(\eps)\di\m(\eps)}{e^{(H\circ {\bf F}^{-1})(E,V;\eps)}-q}
$$
\end{Prop}
\begin{proof}
By assumptions, we can derive the \eqref{ennump} under the symbol of integral obtaining
$$
{\rm det}\begin{pmatrix} \frac{\partial E}{\partial x} &\frac{\partial E}{\partial y}\\\frac{\partial N}{\partial x}& \frac{\partial N}{\partial y}\\\end{pmatrix}
=\int\!\!\!\!\int(\eps_1-\eps_2)\frac{\eps_1 f(\eps_1)g(\eps_2)e^{H(x,y;\eps_1)+H(x,y;\eps_2)}}{\big[\big(e^{H(x,y;\eps_1)}-q\big)\big(e^{H(x,y;\eps_2)}-q\big)\big]^2}\di\m(\eps_1)\di\m(\eps_2)\,.
$$
If such a determinant is non-zero in $(x_0,y_0)$, by the Sign Permanence Theorem it continue to be non-zero in some neighborhood $V\subset U$ of $(x_0,y_0)$. 

The assertion now follows by the Inverse Function Theorem.
\end{proof}
In Section \ref{BZ} below, we will verify by direct computation the situation described above for an explicit model.

\section{the ideal bose gas: the bose-einstein condensation}
\label{0ceb}

According to the results in the previous section (see also \cite{AF}), the generalisation of the Planck formula for the occupation number at energy $\eps$ is given by
\begin{equation}
\label{enb1}
n(\eps)=\frac1{e^{b(f(\eps)\eps-mg(\eps))}-q}\,,
\end{equation}
with $f,g:[0,+\infty)\to[0,+\infty)$ are measurable functions with ${\rm essinf} f,{\rm essinf} g\geq0$, and the ${\rm essinf}$ is defined w.r.t the measure on $[0,+\infty)$ determined by the resolution of the identity $e(\eps)$ in \eqref{1en112} of the hamiltonian. In all cases treated in the present paper, such a measure will be equivalent to the Lebesgue measure. We also assume $b>0$. In the standard thermodynamics, this assumption corresponds to the positivity of the temperature, even if negative ones can appear when the hamiltonian is a bounded operator, see Section 73 of \cite{LL}.

The first step is to determine the admissible values of the generalised chemical potential $m$. As
$$
n\equiv\int g(\eps)n(\eps)\di\m(\eps),\,e\equiv\int f(\eps)\eps n(\eps)\di\m(\eps)<+\infty
$$
is automatically satisfied by assumption, in our setting such conditions leads to
$$
n(\eps)\geq0\,a.e.\,,\quad N\equiv\int n(\eps)\di\m(\eps),\,E\equiv\int \eps n(\eps)\di\m(\eps)<+\infty\,.
$$
Concerning the Fermi-like and the Boltzmann cases, the first condition 
$n(\eps)\geq0$ almost everywhere, is automatically satisfied. Thus we assume that the mean internal energy and particle content of the system is finite, obtaining that for $q\in[-1,0]$ all values of the generalised chemical potential are allowed.

The cases $q\in(0,1]$, including the Bose case $q=1$ are more difficult, see e.g. Section 7 of \cite{AF}.
To manage this situation, we discuss only the Bose case and restrict slightly the choice of the functions $f$ and $g$. 

Without loss of generality, we can assume that
\begin{itemize}
\item[(i)] ${\rm essinf}_{[0,+\infty)}\frac{\eps f(\eps)}{g(\eps)}=0$.
\end{itemize}
In addition, in order to avoid extremely pathological situations, we assume that
\begin{itemize}
\item[(ii)] for each $\a>0$ there exists a measurable non negligible set $A$, and $\d>0$ such that $\eps\in A$ implies $\eps f(\eps)\leq\a g(\eps)$ and $g(\eps)\geq\d$.
\end{itemize}
\begin{Prop}
\label{mall}
Under (i) and (ii) above, we get $m\leq0 \Rightarrow f(\eps)\eps-mg(\eps)\geq0$ almost everywhere, and if $m>0$ there exists a measurable non negligible set $A$ such that 
$f(\eps)\eps-mg(\eps)<0$ on $A$.
\end{Prop}
\begin{proof}
The first assertion is trivial. Concerning the second one, for each $m>0$ choose $\a=m/2$. On the corresponding non negligible set $A$ and for the corresponding $\d>0$, we get
$$
f(\eps)\eps-mg(\eps)\leq-\frac{m}2g(\eps)\leq-\frac{m\d}2<0\,.
$$
\end{proof}
A common phenomenon in the Bose case is the BEC. It is well known that a necessary condition for BEC is that
$\lim_{\eps\to\eps_0}n(\eps)=+\infty$ in \eqref{enb1}, which means 
$$
\lim_{\eps\to\eps_0}\big(f(\eps)\eps-mg(\eps)\big)=0\,.
$$
When this condition is satisfied, a macroscopic amount of particles might occupy the energy level $\eps_0$ in the thermodynamical limit. In the standard thermodynamics, it happens only for $\eps_0=0$ (cf. \cite{BR, Hu, LL}). Notice that in our setting, it was shown in \cite{AF} that the BEC can take place also in excited levels. In the previous mentioned paper, it is seen that the condensation phenomenon can take place also for the exotic cases $q\in(0,1)$.

In order to exhibit states describing BEC even in the more general context of the present paper, we specialise the matter to the simplest model describing non relativistic ideal bosons living on 
$\br^d$. Analogous considerations can be done for bosons on lattices $\bz^d$. Indeed, fix the functions in the class 
$$
\check{\cd}(\br^d)\subset\cs(\br^d)\subset L^2(\br^d,\di^d{\bf x})
$$ 
made of the Fourier Anti-Transform of all infinitely often differentiable functions with compact support in momentum space. The main ingredient will be the opposite of the Laplace operator on 
$\br^d$
$$
-\D=-\sum_{j=1}^d\frac{\partial^2\,\,\,\,}{\partial x^2_j}\,.
$$
The one particle hamiltonian of the model will be $h=(-\D)^{s/2}$, $s\geq1$ avoiding the unphysical cases $s<1$, and including the massive case (with $M=1/2$ for the mass of the particles) 
corresponding to $s=2$, and the phonon/photon hamiltonian (with the speed of the light/sound $c/v=1$) corresponding to $s=1$. We also suppose that the Planck constant is identically 1.
The one particle hamiltonian  is nothing but the multiplication for the function
$$
p^s=\bigg(\sum_{j=1}^dk_j^2\bigg)^{s/2}
$$
in the momentum space after Fourier transform, where as usual, ${\bf p}=(p_1,\dots,p_d)$. 

Proposition \ref{vincc1} asserts that the occupation numbers \eqref{enb1}, on one hand satisfy a maximum principle for the entropy (Microcanonical Ensemble), and on the other hand provide the two point function for a quasi free state satisfying a principle of equilibrium (3.7) of \cite{AF} (corresponding to the Grand Canonical Ensemble in the usual equilibrium), which we report for the convenience of the reader:
\begin{equation}
\label{lepccr2}
\om(a^{\dagger}(\check F)a(\g_{b,m}\check G))=\om(a(\check G)a^{\dagger}(\check F))\,,\quad F,G\in\check{\cd}(\br^d)\,.
\end{equation}
Here, $a^\dagger, a$ are the creator and annihilator distribution satisfying the Canonical Commutation Relations (CCR for short), see \cite{BR}. In addition,
for $b>0$ and all admissible values of $m$ (which leads to $m\leq0$ for the Bose case under the assumptions in Proposition \ref{mall}),
\begin{equation}
\label{crfba}
\om(a^\dagger(\check F)a(\check G)):=
\int_{\br^d}\frac{F({\bf p})\overline{G({\bf p})}}{e^{b(f(p^s)p^s-mg(p^s))}-1}\di^d{\bf p}\,,
\end{equation}
and after Fourier Transform,
\begin{equation*}
(\g_{b,m}F)({\bf p})=e^{b(f(p^s)p^s-mg(p^s))}F({\bf p})\,,\quad F\in{\cd}(\br^d)\,.
\end{equation*}
It is possible that for some admissible values of $m$, also other physically meaningful states satisfy condition \eqref{lepccr2} by the appearance of some delta functions on the r.h.s. of \eqref{crfba}, see \cite{AF, BR}. This is precisely when the BEC takes place, which we are going to discuss.

We start by looking at the local density of particles of the quasi free state $\om$. It is given by
\begin{equation*}
\r_\om({\bf r})=\om(a^\dagger(\d_{\bf r})a(\d_{\bf r}))\,,
\end{equation*}
where $\d_{\bf r}$ is the Dirac distribution centered in ${\bf r}\in\br^d$, provided that the r.h.s. is meaningful, otherwise it is infinite. As it is shown in Section 5 of \cite{AF}, the condensation of particles can take place for any case of $q\in(0,1]$ also in the present setting of non equilibrium thermodynamics.
Here, we report the general Bose case.
\begin{Thm}
\label{tcbb}
For the functions $f,g$ satisfying (i), (ii) above, let for some $b>0$ and $m\leq0$
$$
\frac1{e^{b(f(p^s)p^s-mg(p^s))}-1}\in L^1_{\rm loc}(\br^d)\,,\quad f(x)x-mg(x)\in L^\infty_{\rm loc}(\br_+)\,.
$$
Suppose further that for some $x_0\in[0,+\infty)$, 
$$
\lim_{x\to x_0}\big(f(x)x-mg(x)\big)=0\,.
$$ 
For each bounded positive Radon measure $\n$ on the sphere $\bs_k\subset\br^d$ of radius $k$
and $F,G\in\cd(\br^d)$,
the quasi free state $\om$ with two point function
\begin{equation}
\label{1ab1}
\om(a^\dagger(\check F)a(\check G)):=
\int_{\br^d}\frac{F({\bf p})\overline{G({\bf p})}}{e^{b(f(p^s)p^s-mg(p^s))}-1}\di^d{\bf p}
+\int_{\bs_k}F({\bf p})\overline{G({\bf p})}\di\n({\bf p}) 
\end{equation}
satisfies the condition \eqref{lepccr2}, provided that $k^s=x_0$.
In addition, if 
$$
\frac1{e^{b(f(p^s)p^s-mg(p^s))}-1}\in L^1(\br^d)\,,
$$
then 
\begin{equation*}
\r_\om({\bf r})=\int_{\br^d}\frac{\di^d{\bf p}}{e^{b(f(p^s)p^s-mg(p^s))}-1}+\n(\bs_k)\,,
\end{equation*}
hence it is finite.
\end{Thm}
\begin{proof}
The proof follows {\it mutatis mutandis} the analogous one of Theorem 4.1 of \cite{AF}. We sketch it for the convenience of the reader.

First of all, we note that Proposition \ref{mall} assures that the integrand in \eqref{1ab1} is not negative. Furthermore,
thanks to $\frac1{e^{b(f(p^s)p^s-mg(p^s))}-1}\in L^1_{\rm loc}(\br^d)$, \eqref{1ab1} is well defined for each $F,G\in\cd(\br^d)$. As $f(x)x-mg(x)\in L^\infty_{\rm loc}(\br_+)$ then 
$\g_{b,m}F\in L^2(\br^d)$, provided $F\in\cd(\br^d)$. Denoting by $\cf$ the Fourier Transform, if $k^s=x_0$, the function $\cf\big(e^{b(f(h)h-mg(h))}\check F\big)$ is uniquely defined in ${\bf k}$ as
\begin{align*}
\cf\big(e^{b(f(h)h-mg(h))}\check F\big)&({\bf k})=\big(\lim_{{\bf p}\to{\bf k}}e^{b(f(p^s)p^s-mg(p^s))}\big)F({\bf k})\\
=&\big(\lim_{x\to x_0}e^{b(f(x)x-mg(x))}\big)F({\bf k})=F({\bf k})\,,
\end{align*}
because the function $e^{b(f(p^s)p^s-mg(p^s))}$ coincides a.e. with a measurable function which is continuous in $k^s=x_0$. 

Collecting together, we have that $e^{b(f(h)h-mg(h))}\check f$ is in the domain of the form \eqref{1ab1}. In addition, by using the Canonical Commutation Relations, we compute
\begin{align*}
&\om(a(\check G)a^\dagger(\check F))\\
=&\int_{\br^d}\bigg(1+\frac1{e^{b(f(p^s)p^s-mg(p^s))}-1}\bigg)
F({\bf p})\overline{G({\bf p})}\di^d{\bf p}+\int_{\bs_k}F({\bf p})\overline{G({\bf p})}\di\n({\bf p})\\
=&\int_{\br^d}\frac{e^{b(f(p^s)p^s-mg(p^s))}}{e^{b(f(p^s)p^s-mg(p^s))}-1}
F({\bf p})\overline{G({\bf p})}\di^d{\bf p}+\int_{\bs_k}e^{b(f(p^s)p^s-mg(p^s))}F({\bf k})\overline{G({\bf k})}\di\n({\bf p})\\
=&\int_{\br^d}\frac{F({\bf p})\overline{G({\bf p})}}{e^{b(f(p^s)p^s-mg(p^s))}-1}
\di^d{\bf p}+\int_{\bs_k}F({\bf k})\overline{G({\bf k})}\di\n({\bf p})\\
=&\om\big(a^\dagger(\check F)a\big(e^{b(f(h)h-mg(h))}\check G\big)\big)\,,
\end{align*}
that is \eqref{lepccr2} is satisfied.

If $\frac1{e^{b(f(p^s)p^s-mg(p^s))}-1}\in L^1(\br^d)$, then one easily check that the local density of particles $\r_\om({\bf r})$, which does not depend on ${\bf r}$ as the system under consideration is homogeneous, is finite.
\end{proof}
\begin{Rem}${}$\\
\vskip.005cm
\noindent
(i) In this more general context for which also the particles number is not conserved, we can have condensation for different values of the generalised chemical potential, that is those in the subset
$$
\bigg\{m\leq0\mid \lim_{x\to x_0}\big(f(x)x-mg(x)\big)=0,\,\text{for some}\,\, x_0\geq0\bigg\}\,.
$$
\vskip.02cm
\noindent
(ii) The rotation symmetry is spontaneously broken if $x_0>0$ in Theorem \ref{tcbb}. In order to obtain a rotationally invariant quasi free state, it is enough to take in \eqref{1ab1} any rotationally invariant measure  $\n$ on $\bs_k$.\\
\end{Rem}
The two boundedness conditions in Theorem \ref{tcbb} have a different meaning. The second one $f(x)x-mg(x)\in L^\infty_{\rm loc}(\br_+)$, automatically verified in the usual equilibrium setting when $f=g=1$ and $m=0$, might be relaxed case-by-case when one considers concrete examples. 

The first one 
$\frac1{e^{b(f(p^s)p^s-mg(p^s))}-1}\in L^1_{\rm loc}(\br^d)$ has an important physical meaning. First of all, we note that $\frac1{e^{\b p^s}-1}\in L^1_{\rm loc}(\br^d)$ is equivalent to
$\frac1{e^{\b p^s}-1}\in L^1(\br^d)$ in the equilibrium situation. This is nothing but the assumption that the {\it critical density} at inverse temperature $\b$
\begin{equation*}
\r_c(\b):=\int_{\br^d}\frac{\di^d{\bf p}}{e^{\b p^s}-1}
\end{equation*}
is finite. In our more general situation, the generalised critical density depends on both generalised inverse temperature and chemical potential $b,m$. Moreover, it is possible to have states satisfying \eqref{lepccr2} and 
exhibiting BEC with infinite critical density
$$
\int_{\br^d}\frac{\di^d{\bf p}}{e^{b(f(p^s)p^s-mg(p^s))}-1}\,,
$$
provided $\frac1{e^{b(f(p^s)p^s-mg(p^s))}-1}$ is only locally summable. Such states are unphysical as their local density is infinite. A similar phenomenon happens in studying BEC in equilibrium thermodynamics for inhomogeneous systems, but it is of of different nature. In fact, even for such models the critical density can diverge, but the associated quasi free states have finite local density. It is not a contradiction simply because of inhomogeneity. The reader is referred to \cite{F1, F3, F4, FGI} for a detailed treatment of examples exhibiting such a phenomenon.

\section{the ideal boltzmann gas: a toy example}
\label{BZ}

In order to have an idea of what can happen, we consider the model made of the ideal Boltzmann gas consisting of spinless massive monoatomic particles of mass $M$ in a three dimensional space, for which the weight functions in \eqref{vinc} have the form
$$
f({\bf p})=ap^s\,,\quad g(x)=1\,,
$$
with $a>0$, $s>-2$ are fixed parameters. In all such situations, the involved integrals are convergent, thus everything is mathematically meaningful. The case $a=1$ and $s=0$ corresponds to the usual equilibrium one.

We start by defining the integrals
$$
I_{t,s}:=\int_0^{+\infty}\di x x^te^{-x^{2+s}}\,,\quad t\geq0\,,s>-2\,.
$$
We also put
$$
z:=e^{bm}\,, \quad A:=\frac{a}{2M}\,, \quad B:=Ab\,.
$$
By taking into account \eqref{b4}, and \eqref{b5} with $q=0$ for the passage to the continuum, we get for the density of the particles
$$
\di n({\bf p})=Vn({\bf p})\ds^3{\bf p}=zV\exp(-Bp^{2+s})\ds^3{\bf p}\,.
$$
After passing to spherical coordinates and an elementary change of variable, we obtain 
\begin{equation}
\label{toy}
N=\frac{4\pi I_{2,s}Vz}{h^3B^{\frac{3}{2+s}}}\,,\quad E=\frac{4\pi I_{4,s}AVz}{h^3aB^{\frac{5}{2+s}}}\,,
\end{equation}
which can be solved, obtaining
\begin{equation}
\label{toy1}
B=\bigg(\frac{AI_{4,s}N}{aI_{2,s}E}\bigg)^{1+\frac{s}2}\,,\quad z=\bigg(\frac{AI_{4,s}N}{aI_{2,s}E}\bigg)^{\frac{3}2}\frac{h^3N}{4\pi I_{2,s}V}\,,
\end{equation}
Obviously, for the weighted integral of particles we have $n=N$. Concerning the weighted total energy, we get
\begin{equation}
\label{toy2}
e=\frac{4\pi I_{4+s,s}AVz}{h^3B^{\frac{5+s}{2+s}}}\,,
\end{equation}
which by \eqref{toy1} leads to
\begin{equation}
\label{toy3}
e(N,E)=\frac{I_{4+s,s}(2MI_{2,s})^{\frac{s}2}a}{(I_{4,s})^{1+\frac{s}2}}\bigg(\frac{E}{N}\bigg)^{\frac{s}2}E\,,
\end{equation}
We easily note that $E=e$ for $s=0$ and $a=1$. 

The main object in the Microcanonical Ensemble is the entropy $S=S(N,U,V)$ where $U\equiv E$ is the mean internal energy. We compute 
the entropy for the infinite system subjected to the conditions listed above. After the passage to the continuum, \eqref{entc} for the Boltzmann case becomes
$$
S=kV\int\ds^3{\bf p}n({\bf p})(1-\ln n({\bf p}))\,,
$$
which by \eqref{toy2} leads to
$$
S=kN(1-\ln z)+\frac{kB}{A}e(N,E)\,.
$$
By using \eqref{toy1} and \eqref{toy3}, we obtain
\begin{align}
\label{toy4}
&S(N,V,U)=\bigg(1+\frac{I_{4+s,s}}{I_{2,s}}\bigg)kN\\
+&kN\ln\left[\bigg(\frac{4I_{2,s}}{\sqrt{\pi}}\bigg)^{5/2}\bigg(\frac{3\sqrt{\pi}}{8I_{4,s}}\bigg)^{3/2}\frac{V}{N}\bigg(\frac{MU}{3\pi\hbar^2N}\bigg)^{3/2}\right]\nn\,.
\end{align}
As 
$$
I_{2n,0}=\frac{(2n-1)!!}{2^{n+1}}\sqrt{\pi}
$$
(cf. Section 29 of \cite{LL}), for $s=0$ we get the so called Sackur-Tetrode formula (see e.g. (6.62) of \cite{Hu})
$$
S(N,V,U)=\frac52kN+kN\ln\left[\frac{V}{N}\bigg(\frac{MU}{3\pi\hbar^2N}\bigg)^{3/2}\right]\,.
$$
We note that, solving \eqref{toy4} w.r.t. $U$ and differentiating $U(N,S,V)$ w.r.t. to $S$ and $V$, for the usual temperature and pressure we obtain 
$$
T=\bigg(\frac{\partial U}{\partial S}\bigg)_{N,V}=\frac{2U}{3kN}\,,\quad P=-\bigg(\frac{\partial U}{\partial V}\bigg)_{N,S}=\frac{2U}{3V}\,.
$$
We then recover the equation of the state $PV=NkT$ for such a Boltzmann gas which, surprisingly, coincides with the one of an
ideal Boltzmann gas for the usual equilibrium case (i.e. $s=0, a=1$).

We compute the Helmholtz free energy, which is the thermodynamic potential that measures the available work obtainable from a closed thermodynamic system at a constant temperature in the case of the standard equilibrium. By definition (see e.g. \cite{Hu, LL}), it is given by $F:=U-TS$ and should be computed by the standard variables $N, V, T$.
By taking into account that
$$
\frac{1}{T}\equiv\bigg(\frac{\partial S}{\partial U}\bigg)_{N,V}=\frac{3kN}{2U}\Rightarrow U=\frac32NkT\,,
$$
one recovers that
$$
F(N,V,T)=\frac32NkT-TS\bigg(N,V,\frac32NkT\bigg)\,,
$$
obtaining
$$
F(N,V,T)=\bigg(\frac12-\frac{I_{4+s,s}}{I_{2,s}}\bigg)NkT
-NkT\ln\left[\bigg(\frac{4I_{2,s}}{\sqrt{\pi}}\bigg)^{5/2}\bigg(\frac{3\sqrt{\pi}}{8I_{4,s}}\bigg)^{3/2}\frac{V}{N}\bigg(\frac{MkT}{2\pi\hbar^2}\bigg)^{3/2}\right]\,.
$$
As usual, when $s=0$ we get the equilibrium thermodynamic formula
$$
F(N,V,T)=-NkT\ln\bigg[\frac{eV}{N}\bigg(\frac{MkT}{2\pi\hbar^2}\bigg)^{3/2}\bigg]\,.
$$
It remain open the thermodynamic meaning of the free energy in the more general cases described in the present paper.

For the sake of completeness, we compute the grand potential given in our situation as
$$
\Om:=U-TS-\m N=-PV\,,
$$
where $P$ is the pressure. By using the equation of the state $PV=NkT$, we obtain $\Om=-NkT$ which
should be computed in terms of the variables $V,T,\m$.
By \eqref{en112} \eqref{toy1} and \eqref{toy3}, we recover 
\begin{equation}
\label{toy5}
B=\bigg(\frac{I_{4,s}}{(2+s)MI_{4+s,s}}\frac1{kT}\bigg)^{1+\frac{s}2}\,,\quad z=e^{\frac{\m}{kT}-\frac{3s}{2(1+s)}}\,.
\end{equation}
Inserting \eqref{toy5} in \eqref{toy}, we get
$$
\Om(V,T,\m)=-\frac{4I_{2,s}}{\sqrt{\pi}}\bigg(\frac{(2+s)I_{4+s,s}}{2I_{4,s}}\frac{M}{4\pi\hbar^2}\bigg)^{3/2}(kT)^{5/2}Ve^{\left(\frac{\m}{kT}-\frac{3s}{2(1+s)}\right)}\,.
$$

\section{outlook} 

Motivated by the fact that the (inverse) temperature might be a function of the energy levels in the Planck distribution, we have shown that it can be naturally achieved by imposing the constraint concerning the conservation of a weighted sum of the contributions of the single energy levels occupation in the Microcanonical Ensemble scheme. By supposing that also the number of particles is not conserved, but only a weighted sum of the occupation numbers, we have proposed a way to manage the thermodynamic properties of such open systems. 

Such open systems could be in principle reproducible in laboratory by considering a free gas which can interact with the walls of the boxes by absorbing or relaxing energy according with the energy of the single particle. In addition, the walls could absorb or release particles of the same nature of those present in the box proportionally to the energy levels of the system .

The situation described above might happen in nature where some metastable state can be created in extremely condensed systems far from the usual equilibrium. Such cases might appear in cosmological setting (cf. \cite{V1, V2, V3, V4}), and in plasma physics where a part of the system, small compared with the whole, can exchange matter and energy with the environment according to the rule explained above. The resulting, probably metastable, state of the small system falls in the class of the so called Non Equilibrium Steady States, and it might be managed according to the principles of Irreversible Thermodynamics.  Some similar situation might happen in a rough approximation to the gas of neutrons in nuclear reactors in forming ${}^{239}$Pu
from ${}^{238}$U. 

The proposal contained in the present paper asserts that the cases listed above, more general than the usual ones arising from equilibrium thermodynamics, can still be understood in terms of equilibrium. Indeed, the proposal of the present paper can be summarised as follows. 
\begin{itemize}
\item The scheme to treat the statistical properties of very complex systems having the size of the order of many times the Avogadro Number $N_A$ is the Microcanonical Ensemble where the main object is the Entropy Functional, given in \eqref{entc} for the gas of particles obeying to the various statistics (i.e. Boltzmann, the alternative Bose/Fermi, and formally the $q$-statistics).
\item Within the Microcanonical Ensemble scheme, we can also describe complex systems which can exchange matter and energy with the environment, that is open systems. The isolated systems can be viewed as a particular situation where nothing is exchanged.
\item The equilibrium still assumes the same form: it is reached when the entropy is maximum according to certain constraints imposed to the system like those in \eqref{vinc}, even if more general ones can be considered as well. The isolated systems leading to the standard equilibrium situation corresponds to the case when the weight functions $f,g$ are identically 1.
\end{itemize}
As it has been already discussed in \cite{AF}, there are also some interesting consequences concerning the BEC for the Bose particles (and formally for the exotic particles obeying the $q$-statistics for $0<q<1$). In this new scheme, it is allowed also on excited levels: {\it a non trivial amount of particles can occupy a level for which the energy of such a level is greater than zero}. Indeed by Theorem \ref{tcbb} (see also Theorem 4.1 in \cite{AF}), for the particular situation describing BEC for closed systems, the 
distribution of particles is given by
\begin{equation}
\label{appbec}
\di n(\eps)\propto \bigg[a\d(\eps-\eps_0)+\frac1{e^{bf(\eps)\eps}-1}\bigg]\di\m(\eps)\,,
\end{equation}
where $\di\m(\eps)=h(\eps)\di\eps$ for most of the cases of interest, and $a\d(\eps-\eps_0)$ takes into account of the condensate occupying the level $\eps_0$ satisfying $f(\eps_0)\eps_0=0$. 

In the usual situation when $f=1$, the BEC implies that the amount of condensed can occupy only the level $\eps=0$. Then it seems that the condensation of photons is forbidden in the usual equilibrium thermodynamics: $\eps=0$ implies $\hbar {\bf k}=0$ that is the photon is at rest, which is impossible. Conversely to the standard equilibrium case, we note that the BEC of photons is allowed in the scheme first proposed in \cite{AF} and generalised in Section \ref{0ceb} of the present paper, providing potential applications to nonlinear optics.

On the other hand, very recently such a photon condensation has been detected in the so called "optical microcavity", see \cite{dL, KSVW}. Due to the combined effect of the paraxial approximation $k_z>>k_r$ ($k_r=\sqrt{k_x^2+k_y^2}$ being the modulus of the radial component of the wave vector ${\bf k}=(k_x,k_y,k_z)$) and the trapping due to the mirror curvature of the walls of the cavity, the system is equivalent to one described by the hamiltonian
\begin{align*}
H=\hbar c\sqrt{k_z^2+k_r^2}+\frac{\hbar k_z\Om^2r^2}{2c}\approx&\frac{k_z}{\hbar c} c^2+\frac{(\hbar k_r)^2}{2\frac{k_z}{\hbar c}}+\frac12\frac{k_z}{\hbar c}\Om^2r^2\\
\equiv&M_\text{eff}\,c^2+\frac{(\hbar k_r)^2}{2M_\text{eff}}+\frac12M_\text{eff}\Om^2r^2\,.
\end{align*}
Here, $c\approx$ 300,000 km/sec is the speed of the light and $M_\text{eff},\Om$ are determined by the characteristics of the apparatus.

Roughly speaking, the condensation of photons in such a model is allowed because the portion of condensate have only the axial component of the momentum (i.e. $k_r=0$ for the condensate), then without violating 
$$
p\equiv\hbar\sqrt{k_x^2+k_y^2+k_z^2}\neq0\,.
$$
Such a condensation effect might take place also at the cosmological level, playing possibly a role in the explanation of the dark energy of the universe.

The main consequence of the model briefly outlined above is the appearance of a "effective mass" to be determined experimentally. Apart from the fact that we can have BEC on excited levels, we also notice that the appearance of an effective mass is easily allowed in our model as well. In fact, for the one particle photon hamiltonian $\hbar c k$, it is enough to chose $f({\bf k})=a k$, obtaining in the Planck distribution \eqref{vinc1},
$$
\hbar cf({\bf k})k=\frac{(\hbar k)^2}{2M_\text{eff}}\,.
$$
where $M_\text{eff}=\frac{\hbar}{2ac}$.

In the present paper, we have provided a new proposal for the thermodynamics of enormous open systems which can be classified under the name of {\it irreversible thermodynamics}. This can open the perspective, out of the aim of the present paper, of investigating both from a theoretical viewpoint and for concrete models, the arising thermodynamical consequences corresponding to the analogous ones of the usual equilibrium thermodynamics.

\section{appendix: the continuum limit}
\label{appp}

Even if the analysis of Section \ref{1mic1} is correct, it is not applicable to the majority of concrete examples for which the involved hamiltonians have continuous spectrum. In this situation, most of the explicit computations can be carried out only for the ideal gas made of identical particles (or composed systems whose parts are still made of ideal gases which do no interact each other). However, these cases are the mostly treated in literature, being very important also from the conceptual viewpoint. Hence, for the convenience of the reader we outline and justify the passage to the continuum limit for the formulas in Sections \ref{1mic1} and \ref{tprr} in the case of the ideal gas. 

To start with, as usual (e.g. Section 8.5 of \cite{Hu}), we consider a free gas in a rectangular box of size $L$ in $d$ dimensional space. As well known, the momentum 
${\bf p}$ of any particle composing the gas is given by 
${\bf p}=\frac{2\pi\hbar{\bf n}}{L}$, where $\hbar:=h/2\pi$, $h\approx 6.626070040\times 10^{-34} Js$ being the Planck constant, and ${\bf n}$ denotes the numbers occupation vector:
${\bf n}=(n_1,n_2,\cdots, n_d),\;n_k=0,\pm 1,\pm 2,\cdots$. For the energy 
$\epsilon(\{n_k\})$ of a non relativistic ideal gas composed of particles with mass $M$ we have:
\begin{equation}
\epsilon(n_1,n_2,\cdots,n_d)=\frac{p^2}{2M}=\frac{{\hbar}^d}{2M}
{\left(\frac{2\pi}{L}\right)}^d(n_1^2+n_2^2+\cdots+n_d^2)\,.
\label{b1}
\end{equation}
The number of the allowed states $\Gamma$ in the intervals $\Delta n_1,\Delta n_2,\cdots,\Delta n_d$ is thus given by 
$\Gamma=\Delta n_1\Delta n_2\cdots\Delta n_d$.
Thanks to \eqref{b1}, we get
$\Gamma=\frac{L^d}{h^d}\Delta p_1\Delta p_2\cdots\Delta p_d$.
We are now in the position to build the continuum limit of formulae \eqref{0entc}, \eqref{vinc} and \eqref{en02}, obtaining for the ideal gases the usual results (see e.g. \cite{Hu, LL}), and those in Section \ref{BZ}.

First of all, quantum effects come in action
when the allowed states are not "very" near, i.e. when the value for the action of the system is comparable with $h$. As an example, for 
$L\sim 10^{-8}\;cm$ quantum effects cannot be neglected, while for $L>>10^{-8}\;cm$ the classical tractation is appropriate.
From \eqref{b1}, we see that the summation over '$i$' in \eqref{0entc}, \eqref{vinc} and \eqref{en02} can be replaced by a summation over all values set of $\{n_k\}$. In this regard, we have
\begin{align*}
&S_q(\{\n_i\}):=\sum_{\{n_k\}}\bigg[\frac{1+q\bar\n(\{n_k\})}{q}\ln(1+q\bar\n(\{n_k\}))-\bar\n(\{n_k\})\ln\bar\n(\{n_k\})\bigg]\,,\\
&E=\sum_{\{n_k\}}\eps(\{n_k\})\bar\n(\{n_k\})\,,\quad N=\sum_{\{n_k\}}\bar\n(\{n_k\})\,,\\
&e=\sum_{\{n_k\}}f(\eps(\{n_k\}))\eps(\{n_k\})\bar\n(\{n_k\})\,,\quad n=\sum_{\{n_k\}}g(\eps(\{n_k\}))\bar\n(\{n_k\})\,,
\end{align*}
where $\bar\n(\{n_k\})$ is given by \eqref{vinc1}:
\begin{equation}
\label{b3}
\bar\n(\{n_k\})=\frac1{e^{b[f(\eps(\{n_k\}))\eps(\{n_k\})-mg(\eps(\{n_k\}))]}-q}.
\end{equation} 
In this way, the degeneracy factor $g_i$ in  \eqref{0entc}, \eqref{vinc} and \eqref{en02} is absorbed into the summation over $\{n_k\}$.
As a consequence, the continuum limit is now obtained in the following simple way. For a gas with $N$ particles, we should make the replacements
\begin{equation}
\sum_{i}g_i\rightarrow \sum_{\{n_k\}}\rightarrow
V_d\int\frac{\di^d{\bf p}}{h^d}\,,
\label{b4}
\end{equation}
and by \eqref{b3},
\begin{equation}
\bar\n(\{n_k\})\rightarrow n({\bf p}):=\frac1{e^{b[f(\eps({\bf p}))\eps({\bf p})-mg(\eps({\bf p}))]}-q}
\label{b5}
\end{equation}
for the occupation numbers. Accordingly, to simplify notations we define the differential $\ds^d{\bf p}:=\frac{\di^d{\bf p}}{h^d}$ having the dimension $[\text{length}]^{-d}$ of a wave vector in the $d$ dimensional space.
With such a natural normalisation, the $2^{nd}$ part of \eqref{en02} becomes
$$
\int\frac{\ds^d{\bf p}}{e^{b[f(\eps({\bf p}))\eps({\bf p})-mg(\eps({\bf p}))]}-q}=\frac{N}{V_d}\,.
$$
Namely, the generalised Planck distribution $n({\bf p})$ in \eqref{b5} assumes the meaning of the spatial density of the particles having momenta around ${\bf p}=(p_1,\dots,p_d)$ in the infinitesimal hypercube of volume $\di p_1/h\dots\di p_d/h$.

\section*{Acknowledgement} 

The first named author is partially supported by Italian INDAM-GNAMPA. The authors kindly acknowledge L. Accardi for many fruitful discussions on the Local KMS Principle and its potential applications.

\end{document}